\documentclass[runningheads,a4paper]{llncs}
\usepackage{amssymb}
\usepackage{graphicx}
\usepackage{url}
\usepackage{multirow}
\usepackage{mathtools}
\usepackage{color}
\usepackage{here}
\usepackage{stackrel}
\usepackage[margin=1.64in]{geometry}
\usepackage{amsmath}
\usepackage{lineno}
\usepackage{thm-restate}

\urldef{\mailsa}\path|{x}@inf.ethz.ch|
\newcommand{\keywords}[1]{\par\addvspace\baselineskip
\noindent\keywordname\enspace\ignorespaces#1}
\spnewtheorem{observation}{Observation}{\bfseries}{\itshape}

\begin{document}
\mainmatter
\title{Tight Bounds on the Minimum Size of a Dynamic Monopoly}

\titlerunning{Tight Bounds on the Minimum Size of a Dynamic Monopoly}
\author{
    Ahad N. Zehmakan
}
\authorrunning{Ahad N. Zehmakan}
\institute{
Department of Computer Science, ETH Z\"urich, Switzerland\\
\email{abdolahad.noori@inf.ethz.ch}
}
\toctitle{}
\tocauthor{}
\maketitle
\begin{abstract}
Assume that you are given a graph $G=(V,E)$ with an initial coloring, where each node is black or white. Then, in discrete-time rounds all nodes simultaneously update their color following a predefined deterministic rule. This process is called two-way $r$-bootstrap percolation, for some integer $r$, if a node with at least $r$ black neighbors gets black and white otherwise. Similarly, in two-way $\alpha$-bootstrap percolation, for some $0<\alpha<1$, a node gets black if at least $\alpha$ fraction of its neighbors are black, and white otherwise. The two aforementioned processes are called respectively $r$-bootstrap and $\alpha$-bootstrap percolation if we require that a black node stays black forever.

For each of these processes, we say a node set $D$ is a dynamic monopoly whenever the following holds: If all nodes in $D$ are black then the graph gets fully black eventually. We provide tight upper and lower bounds on the minimum size of a dynamic monopoly.
\end{abstract}
\section{Introduction}
Suppose for a graph $G$ by starting from an initial configuration (coloring), where each node is either black or white, in each round all nodes simultaneously update their color based on a predefined rule. This basic abstract model, which is commonly known as cellular automaton, has been studied extensively in different areas, like biology, statistical physics, and computer science to comprehend the behavior of various real-world phenomena. 

In two-way $r$-bootstrap percolation (or shortly $r$-BP), for some positive integer $r$, in each round a node gets black if it has at least $r$ black neighbors, and white otherwise. Furthermore, in two-way $\alpha$-bootstrap percolation ($\alpha$-BP), for some $0<\alpha<1$, a node gets black if at least $\alpha$ fraction of its neighbors are black, and white otherwise. (Notice that there should not be any confusion between two-way $\alpha$-BP and two-way $r$-BP since $r$ is an integer value larger than equal to 1 and $0<\alpha<1$.) These two processes are supposed to model social phenomena like opinion forming in a community, where black and white could represent respectively positive and negative opinion concerning a reform proposal or a new product. For instance in a social network, if a certain number/fraction of someone's connections have a positive opinion regarding a particular topic, s/he will adapt the same opinion, and negative otherwise.

$r$-BP and $\alpha$-BP are defined analogously, except we require that a black node stays unchanged. The main idea behind these two variants is to model monotone processes like rumor spreading in a society, fire propagation in a forest, and infection spreading among cells. For example, an individual gets informed of a rumor if a certain number/fraction of his/her friends are aware of it and stays informed forever, or a tree starts burning if a fixed number/fraction of the adjacent tress are on fire.

If we can try to convince a group of individuals to adopt a new product or innovation, for instance by providing them with free samples, and the goal is to trigger a large cascade of further adoptions, which set of individuals should we target and how large this set should be? This natural question brings up the well-studied concept of a dynamic monopoly. For each of the above four models, we say a node set is a dynamic monopoly, or shortly dynamo, whenever the following holds: If all nodes in the set are black initially then all nodes become black eventually.

Although the concept of a dynamo was studied earlier, e.g. by Balogh and Pete~\cite{balogh1998random}, it was formally defined and studied in the seminal work by Kempe, Kleinberg, and Tardos~\cite{kempe2003maximizing} and independently by Peleg~\cite{peleg1998size}, respectively motivated from viral marketing and fault-local mending in distributed systems. There is a massive body of work concerning the minimum size of a dynamo in different classes of graphs, for instance hypercube~\cite{balogh2010bootstrap,morrison2018extremal}, the binomial random graph~\cite{janson2012bootstrap,feige2017contagious,chang2012triggering,n2018opinion}, random regular graphs~\cite{balogh2007bootstrap,gartner2018majority}, and many others. Motivated from the literature of statistical physics, a substantial amount of attention has been devoted to the $d$-dimensional lattice, for instance see~\cite{balister2010random,flocchini2004dynamic,gartner2017color,gartner2017biased,jeger2018dynamic,zehmakan2018two}.

In the present paper, we do not limit ourselves to a particular class of graphs and aim to establish sharp lower and upper bounds on the minimum size of a dynamo in general case, in terms of the number of nodes and the maximum/minimum degree of the underlying graph. See Table~\ref{Table 1} for a summary.  

Some of the bounds are quite trivial. For example in (two-way) $r$-BP, $r$ is an obvious lower bound on the minimum size of a dynamo and it is tight since the complete graph $K_n$ has a dynamo of size $r$ (see Lemma~\ref{lemma} for the proof). However, some of the bounds are much more involved. For instance, an interesting open problem in this literature is whether the minimum size of a dynamo in two-way $\alpha$-BP for $\alpha>1/2$ is bounded by $\Omega(\sqrt{n})$ or not. We prove that this is true for $\alpha>\frac{3}{4}$. The case of $\frac{1}{2}<\alpha\le\frac{3}{4}$ is left for the future research.

The proof techniques utilized are fairly standard and straightforward (some new and some inspired from prior work), however they turn out to be very effective. The upper bounds are built on the probabilistic method. We introduce a simple greedy algorithm which always returns a dynamo. Then, we discuss if this algorithm visits the nodes in a random order, the expected size of the output dynamo matches our desired bound. For the lower bounds, we define a suitable potential function, like the number of edges whose endpoints have different colors in a configuration. Then, careful analysis of the behavior of such a potential function during the process allows us to establish lower bounds on the size of a dynamo. To prove the tightness of our results, we provide explicit graph constructions for which the minimum size of a dynamo matches our bounds.

A simple observation is that by adding an edge to a graph the minimum size of a dynamo in (two-way) $r$-BP does not increase (for a formal argument, please see Section~\ref{dynamo}). Thus, if one keeps adding edges to a graph, eventually it will have a dynamo of minimum possible size, i.e. $r$. Thus, it would be interesting to ask for the degree-based density conditions that ensure that a graph $G$ has a dynamo of size $r$. This was studied for $r$-BP, in the terms of the minimum degree, by Freund, Poloczek, and Reichman~\cite{freund2018contagious}. They proved that if the minimum degree $\delta(G)$ is at least $\lceil \frac{r-1}{r}n\rceil$, then there is a dynamo of size $r$ in $G$. Gunderson~\cite{gunderson2017minimum} showed that the statement holds even for $\delta(G) \ge \frac{n}{2}+r$, and this is tight up to an additive constant. We study the same question concerning the two-way variant and prove that if $\delta(G)\ge \frac{n}{2}+r$ then the graph includes $\Omega(n^r)$ dynamos of size $r$. Note that this statement is stronger than Gunderson's result. Firstly, we prove that there is a dynamo of size $r$ in two-way $r$-BP, which implies there is a dynamo of size $r$ in $r$-BP. Moreover, we show that there is not only one but also $\Omega(n^r)$ of such dynamos. It is worth to stress that our proof is substantially shorter and simpler.

We say a dynamo is monotone if it makes all nodes black monotonically, that is no black node ever becomes white during the process. In $r$-BP and $\alpha$-BP, any dynamo is monotone by definition, but in the two-way variants this is not necessarily true. Monotone dynamos also have been studied in different classes of graphs, see e.g.~\cite{flocchini2004dynamic,balister2010random,peleg1998size}. We provide tight bounds on the minimum size of a monotone dynamo in general graphs and also the special case of trees. In particular in two-way $\alpha$-BP for $\alpha>\frac{1}{2}$, we prove the tight lower bound of $\sqrt{\frac{\alpha}{1-\alpha}n}$ on the minimum size of a monotone dynamo in general case. Interestingly, this bound drastically increases if we limit the underlying graph to be a tree. A question which arises, is whether there is a relation among the minimum size of a monotone dynamo and the girth of the underlying graph or not? This is partially answered in~\cite{coja2015contagious} for $r$-BP. 

If all nodes in a dynamo are black, then black color will occupy the whole graph eventually. What if we only require the black color to survive in all upcoming rounds, but not necessarily occupy the whole graph? To address this question, we introduce two concepts of a stable set and immortal set. A non-empty node set $S$ is stable (analogously immortal) whenever the following hold: If initially $S$ is fully black, it stays black forever (respectively, black color survives forever). (See Section~\ref{pre} for formal definitions) Trivially, a stable set is also immortal, but not necessarily the other way around. Similar to dynamo, we provide tight bounds on the minimum size of a stable and an immortal set; see Table~\ref{Table 2}. In $r$-BP and $\alpha$-BP, a black node stays unchanged; thus, the minimum size of a stable/immortal set is equal to one. However, the situation is a bit more involved in two-way variants. Surprisingly, it turns out that in two-way $2$-BP the parity of $n$, the number of nodes in the underlying graph, plays a key role.

The layout of the paper is as follows. First, we set some basic definitions in Section~\ref{pre}. Then, the bounds for dynamos, monotone dynamos, and stable/immortal sets are presented respectively in Sections~\ref{dynamo},~\ref{monotone}, and~\ref{stable}.
\subsection{Preliminaries}
\label{pre}
Let $G=(V,E)$ be a graph that we keep fixed throughout. We always assume that $G$ is connected. For a node $v\in V$, $\Gamma(v):=\{u\in V: \{u,v\} \in E\}$ is the \emph{neighborhood} of $v$. For a set $S\subset V$, we define $\Gamma(S):=\bigcup_{v\in S}\Gamma(v)$ and $\Gamma_S(v):=\Gamma(v)\cap S$. Furthermore, $d(v):=|\Gamma(v)|$ is the \emph{degree} of $v$ and $d_S(v):=|\Gamma_S(v)|$. We also define $\Delta(G)$ and $\delta(G)$ to be respectively the maximum and the minimum degree in graph $G$.
(To lighten the notation, we sometimes shortly write $\Delta$ and $\delta$ where $G$ is clear form the context). 
In addition, for a node set $A\subset V$, we define the \emph{edge boundary} of $A$ to be $\partial(A):=\{\{u,v\}:v\in A\wedge u\in V\setminus A\}$. 

A \emph{configuration} is a function $\mathcal{C}:V\rightarrow\{b,w\}$, where $b$ and $w$ stand for black and white. For a node $v\in V$, the set $\Gamma_a^{\mathcal{C}}(v):=\{u\in \Gamma(v):\mathcal{C}(u)=a\}$ includes the neighbors of $v$ which have color $a\in\{b,w\}$ in configuration $\mathcal{C}$. We write $\mathcal{C}|_S=a$ for a set $S\subseteq V$ and color $a\in\{b,w\}$ if $\mathcal{C}(u)=a$ for every $u\in S$.

Assume that for a given initial configuration $\mathcal{C}_0$ and some integer $r\ge 1$, $\mathcal{C}_t(v)$, which is the color of node $v\in V$ in round $t\ge 1$, is equal to $b$ if $|\Gamma_b^{\mathcal{C}_{t-1}}(v)|\ge r$, and $\mathcal{C}_{t}(v)=w$ otherwise. This process is called \emph{two-way r-bootstrap percolation}. If we require that a black node to stay black forever, i.e., $\mathcal{C}_t(v)=w$ if and only if $|\Gamma_b^{\mathcal{C}_{t-1}}(v)|< r$ and $\mathcal{C}_{t-1}(v)=w$, then the process is called \emph{r-bootstrap percolation}.

\textbf{Assumptions.}  We assume that $r$ is fixed while we let $n$, the number of nodes in the underlying graph, tend to infinity. Note that if $d(v)<r$ for a node $v$, it never gets black in (two-way) $r$-BP, except it is initially black; thus, we always assume that $r\le \delta(G)$.

Furthermore, suppose that for a given initial configuration $\mathcal{C}_0$ and some fixed value $0<\alpha< 1$, $\mathcal{C}_t(v)=b$ for $v\in V$ and $t\ge 1$ if $|\Gamma_b^{\mathcal{C}_{t-1}}(v)|\ge \alpha |\Gamma(v)|$ and $\mathcal{C}_{t}(v)=w$ otherwise. This process is called \emph{two-way $\alpha$-bootstrap percolation}. If again we require a black node to stay unchanged, i.e., $\mathcal{C}_t(v)=w$ if and only if $|\Gamma_b^{\mathcal{C}_{t-1}}(v)|< \alpha |\Gamma(v)|$ and $\mathcal{C}_{t-1}(v)=w$, then the process is called \emph{$\alpha$-bootstrap percolation}.

For any of the above processes on a connected graph $G=(V,E)$, we define a node set $D$ to be a \emph{dynamic monopoly}, or shortly \emph{dynamo}, whenever the following holds: If $\mathcal{C}_{t}|_D=b$ for some $t\ge 0$, then $\mathcal{C}_{t'}|_V=b$ for some $t'\ge t$. Furthermore, assume that for a non-empty node set $S$, if $\mathcal{C}_{t}|_S=b$ for some $t\ge 0$, then $\mathcal{C}_{t'}|_S=b$ for any $t'\ge t$; then, we say $S$ is a \emph{stable set}. Finally, a node set $I$ is an \emph{immortal set} when the following is true: If $\mathcal{C}_{t}|_I=b$ for some $t\ge 0$, then for any $t'\ge t$ there exists a node $v\in V$ so that $\mathcal{C}_{t'}(v)=b$. 

For a graph $G$ we define the following notations:
\begin{itemize}
\item $MD_r(G):=$ The minimum size of a dynamo in $r$-BP.
\item $\overleftarrow{MD}_r(G):=$ The minimum size of a dynamo in two-way $r$-BP.
\item $MS_r(G):=$ The minimum size of a stable set in $r$-BP.
\item $\overleftarrow{MS}_r(G):=$ The minimum size of a stable set in two-way $r$-BP. 
\item $MI_r(G):=$ The minimum size of an immortal set in $r$-BP.
\item $\overleftarrow{MI}_r(G):=$ The minimum size of an immortal set in two-way $r$-BP. 
\end{itemize}
We analogously define $MD_{\alpha}(G)$, $MS_{\alpha}(G)$, $MI_{\alpha}(G)$ for $\alpha$-BP and $\overleftarrow{MD}_{\alpha}(G)$, $\overleftarrow{MS}_{\alpha}(G)$, and $\overleftarrow{MI}_{\alpha}(G)$ for two-way $\alpha$-BP.


As a warm-up, let us compute some of these parameters for some specific class of graphs in Lemma~\ref{lemma}, which actually come in handy several times later for arguing the tightness of our bounds.
\begin{lemma}
\label{lemma}
For complete graph $K_n$, $MD_r(K_n)=\overleftarrow{MD_r}(K_n)=r$ and $MD_{\alpha}(K_n)\ge\lceil\alpha n\rceil-1$. Furthermore, for an $r$-regular graph $G=(V,E)$ and $r\ge 2$, $\overleftarrow{MD_r}(G)=n$.
\end{lemma}
\begin{proof}
Firstly, $r\le MD_r(K_n)$ because by starting from a configuration with less than $r$ black nodes in $r$-BP, clearly in the next round all nodes will be white. Secondly, $\overleftarrow{MD_r}(K_n)\le r$ because from a configuration with $r$ black nodes in two-way $r$-BP, in the next round all the $n-r$ white nodes turn black and after one more round all nodes will be black because $n-r$ is at least $r+1$ (recall that we assume that $r$ is fixed while $n$ tends to infinity). By these two statements and the fact that $MD_r(K_n)\le\overleftarrow{MD_r}(K_n)$ (this is true since a dynamo in two-way $r$-BP is also a dynamo in $r$-BP), we have $MD_r(K_n)=\overleftarrow{MD_r}(K_n)=r$. In two-way $\alpha$-BP on $K_n$, by starting with less than $\lceil \alpha (n-1)\rceil$ black nodes, the process gets fully white in one round, which implies that $MD_{\alpha}(K_n)\ge\lceil\alpha n\rceil-1$. (The interested reader might try to find the exact value of $MD_{\alpha}(K_n)$ as a small exercise.) 

For two-way $r$-BP on an $r$-regular graph $G$, consider an arbitrary configuration with at least one white node, say $v$. Trivially, in the next round all nodes in $\Gamma(v)$ will be white. Thus, by starting from any configuration except the fully black configuration, the process never gets fully black. This implies that $\overleftarrow{MD}_r(G)=n$. (We exclude $r=1$ because a $1$-regular graph is disconnected for large $n$.)
\qed
\end{proof}

\section{Lower and Upper Bounds}
\subsection{Dynamos}
\label{dynamo}
In this section, we provide lower and upper bounds on the minimum size of a dynamo in $\alpha$-BP (Theorem~\ref{thm1}), two-way $\alpha$-BP (Theorem~\ref{thm2}), $r$-BP (Theorem~\ref{thm3}), and two-way $r$-BP (Theorem~\ref{thm4}). See Table~\ref{Table 1} for a summary. Furthermore in Theorem~\ref{thm5}, we present sufficient minimum degree condition for a graph to have a dynamo of size $r$ in two-way $r$-BP.
\vspace{-0.4cm}
\begin{table}
\centering
\caption{The minimum size of a dynamo. All bounds are tight up to an additive constant, except some of the bounds for two-way $\alpha$-BP.}
\label{Table 1}
\begin{tabular}{ |c|c|c| } 
 \hline
Model& Lower Bound &Upper Bound\\
 \hline
 $\alpha$-BP   &1 & $(\frac{\delta+\frac{1}{\alpha}}{\delta+1})\ \alpha n$\\
 \hline
 Two-way $\alpha$-BP $\ \alpha>\frac{3}{4}$&   $2\alpha\sqrt{n}-1$  & $n$\\
 \hline
 Two-way $\alpha$-BP $\ \alpha\le\frac{3}{4}$&   $1$  & $n$\\
 \hline
 $r$-BP   & $r$   &$(\frac{r}{1+\delta})\ n$\\
 \hline
 Two-way $r$-BP$\ r\ge 2$&   $r$  &$n$ \\
 \hline
 Two-way $r$-BP $\ r=1$ & 1 & 2\\
 \hline
\end{tabular}
\end{table}
\vspace{-0.5cm}
\begin{theorem}
\label{thm1}
For a graph $G=(V,E)$, $1\le MD_{\alpha}(G)\le (\frac{\delta+\frac{1}{\alpha}}{\delta+1})\ \alpha n$. 
\end{theorem}
In Theorem~\ref{thm1}, the lower bound is trivial and the upper bound is proven by applying an idea similar to the one from Theorem 2.1 in~\cite{reichman2012new}.
\begin{proof}
By applying the probabilistic method, we prove the upper bound. Consider an arbitrary labeling $L:V\rightarrow [n]$, which assigns a unique label from $1$ to $n$ to each node. Define the set $D_L:=\{v \in V:|\{u\in\Gamma(v):L(u)<L(v)\}|<\alpha d(v)\}$. We claim that $D_L$ is a dynamo in $\alpha$-BP, irrespective of $L$. More precisely, we show that by starting from a configuration where $D_L$ is fully black, in the $t$-th round for $t\geq 1$ all nodes with label $t$ or smaller are black. This immediately implies that $D_L$ is a dynamo since in at most $n$ rounds the graph gets fully black. The node with label $1$ is in $D_L$ by definition; thus, it is black in the first round. As the induction hypothesis, assume that all nodes with label $t$ or smaller are black in the $t$-th round for some $t\ge 1$. If node $v$ with label $t+1$ is in $D_L$ then it is already black; otherwise, it has at least $\alpha d(v)$ neighbors with smaller labels, which are black by the induction hypothesis. Thus, $v$ will be black in the $(t+1)$-th round and all nodes with smaller labels also will stay black.

Assume that we choose a labeling $L$ uniformly at random among all $n!$ possible labellings. Let us compute the expected size of $D_L$.
\begin{align*}
& \mathbb{E}[|D_L|]= \sum_{v\in V} Pr[\text{node $v$ is in $D_L$}]=\sum_{v\in V}\frac{\lceil \alpha d(v)\rceil}{d(v)+1}\le \sum_{v\in V}\frac{\alpha d(v)+1}{d(v)+1}\\
& \le\sum_{v\in V}\frac{\alpha\delta+1}{\delta+1}=(\frac{\delta+\frac{1}{\alpha}}{\delta+1})\alpha n.
\end{align*}
Therefore, there exists a labeling $L$ with $|D_L|\le (\frac{\delta+\frac{1}{\alpha}}{\delta+1})\alpha n$, which implies that there exists a dynamo of this size. \qed
\end{proof}
\textbf{Tightness.} The lower bound is tight since in the star $S_n$, a tree with one internal node and $n$ leaves, the internal node is a dynamo, irrespective of $0<\alpha<1$. Furthermore for the complete graph $K_n$, $MD_{\alpha}(K_n)\ge \lceil \alpha n\rceil-1$ (see Lemma~\ref{lemma}) and our upper bound is equal to $\alpha n+1-\alpha$ (by plugging in the value $\delta=n-1$). Thus, the upper bound is tight up to an additive constant.
\begin{theorem}
\label{thm2}
For a graph $G=(V,E)$,

(i) $2\alpha\sqrt{n}-1\le \overleftarrow{MD}_{\alpha}(G)\le n$ for $\alpha>\frac{3}{4}$

(ii) $1\le \overleftarrow{MD}_{\alpha}(G)\le n$ for $\alpha\le \frac{3}{4}$.
\end{theorem}
\begin{proof}
We prove the lower bound of $2\alpha\sqrt{n}-1$; all other bounds are trivial. Let $D$ be an arbitrary dynamo in two-way $\alpha$-BP for $\alpha>\frac{3}{4}$. Consider the initial configuration $\mathcal{C}_0$ where $\mathcal{C}_0|_D=b$ and $\mathcal{C}_0|_{V\setminus D}=w$. Furthermore, we define for $t\ge 1$ $B_t:=\{v\in V: \mathcal{C}_{t'-1}(v)=\mathcal{C}_{t'}(v)=b\ \text{for some}\ t'\le t\}$
to be the set of nodes which are black in two consecutive rounds up to the $t$-th round.

Now, we define the potential function $\Phi_t:=|B_t|+|\partial(B_t)|$ to be the number of nodes in $B_t$ plus the number of edges with exactly one endpoint in $B_t$. Since $D$ is a dynamo, there exists some $T\ge 1$ such that $\mathcal{C}_T|_V=b$, which implies that $\mathcal{C}_{T+1}|_V=b$. Thus, we have $\Phi_{T+1}=|B_{T+1}|+|\partial(B_{T+1})|=n+0=n$. We prove that $\Phi_1\le \frac{1}{4\alpha^2}(|D|+(2\alpha-1))^2$ and $\Phi_{t+1}\le \Phi_{t}$ for any $t\ge 1$. Therefore, $n=\Phi_{T+1}\le \Phi_1\le \frac{1}{4\alpha^2}(|D|+(2\alpha-1))^2$, which results in
\[
4\alpha^2 n\le (|D|+(2\alpha-1))^2\Rightarrow 2\alpha\sqrt{n}-2\alpha+1\le |D|\xRightarrow{2\alpha< 2} 2\alpha\sqrt{n}-1\le |D|.
\]
Firstly, we prove that $\Phi_{t+1}\le \Phi_{t}$ for any $t\ge 1$. Define the set $B:=B_{t+1}\setminus B_t$. If $B=\emptyset$, then $B_t=B_{t+1}$, because by definition $B_t\subseteq B_{t+1}$, which implies that $\Phi_{t+1}=\Phi_{t}$. Thus, assume that $B\ne\emptyset$. A node $v\in B$ is black in both rounds $t$ and $t+1$, which means it has at least $\alpha d(v)$ black neighbors in the $(t-1)$-th round and $\alpha d(v)$ black neighbors in the $t$-th round. Thus by the pigeonhole principle, at least $2\alpha-1$ fraction of its neighbors are black in both rounds $t-1$ and $t$. In other words, at least $2\alpha-1$ fraction of its neighbors are in $B_t$. Note that $2\alpha-1>\frac{1}{2}$ for $\alpha>\frac{3}{4}$. Therefore, for each node $v\in B$ more than half of its neighbors are in $B_t$. This implies that $|\partial(B_{t+1})|\le |\partial(B_t)|- |B|$. Thus,
\[
\Phi_{t+1}=|B_{t+1}|+|\partial(B_{t+1})|\le |B_t|+|B|+|\partial(B_t)|-|B|=|B_t|+|\partial(B_t)|=\Phi_t.
\]
It remains to show that $\Phi_1\le \frac{1}{4\alpha^2}(|D|+(2\alpha-1))^2$. Recall that for a node $v\in B_1$, $d_{V\setminus D}(v)$ and $d_{D\setminus B_1}(v)$ are the number of edges that $v$ shares with nodes in $V\setminus D$ and $D\setminus B_1$, respectively. In addition, note that $\mathcal{C}_0(v)=\mathcal{C}_1(v)=b$ by the definition of $B_1$. Since all nodes in $V\setminus D$ are white in $\mathcal{C}_0$ and $v$ must be black in $\mathcal{C}_1$, $\frac{d_{V\setminus D}(v)}{d(v)}\le (1-\alpha)$. Furthermore, since $v$ has at most $|D|-1$ neighbors in $D$ (this is true because $\mathcal{C}_0(v)=b$, that is $v\in D$), we have that $|D|-1+d_{V\setminus D}(v)\ge d(v)$. Thus, $\frac{d_{V\setminus D}(v)}{|D|-1+d_{V\setminus D}(v)}\le (1-\alpha)$, which implies that $d_{V\setminus D}(v)\le \frac{1-\alpha}{\alpha}(|D|-1)$. Moreover, since $B_1\subseteq D$, $d_{D\setminus B_1}(v)\le |D|-|B_1|$. Putting the last two statements together outputs $d_{V\setminus B_1}(v)\le |D|-|B_1|+\frac{1-\alpha}{\alpha}(|D|-1)=\frac{1}{\alpha}|D|-|B_1|+\frac{\alpha-1}{\alpha}$. Therefore 
\begin{align*}
& \Phi_1= |B_1|+|\partial(B_1)|\le|B_1|+ |B_1|\cdot (\frac{1}{\alpha}|D|-|B_1|+\frac{\alpha-1}{\alpha})=\\
& (\frac{1}{\alpha}|D|+\frac{2\alpha-1}{\alpha})|B_1|-|B_1|^2=\frac{|D|+(2\alpha-1)}{\alpha}|B_1|-|B_1|^2.
\end{align*}
The upper bound is maximized for $|B_1|=\frac{1}{2}(\frac{|D|+(2\alpha-1)}{\alpha})$. Thus, $\Phi_1\le \frac{1}{4\alpha^2}(|D|+(2\alpha-1))^2$.
\qed  
\end{proof}
\textbf{Tightness.} Let us first consider part (i). We show that there exist $n$-node graphs with dynamos of size $k=\sqrt{\frac{\alpha}{1-\alpha}n}$ for $\alpha> \frac{3}{4}$ (actually our construction works also for $\alpha\ge 1/2$), which demonstrates that our bound is asymptotically tight. Consider a clique of size $k$ and attach $\frac{n}{k}-1$ distinct leaves to each of its nodes. The resulting graph has $k+k(\frac{n}{k}-1)=n$ nodes. Consider the initial configuration $\mathcal{C}_0$ in which the clique is fully black and all other nodes are white. In $\mathcal{C}_1$, all the leaves turn black because their neighborhood is fully black in $\mathcal{C}_0$. Furthermore, each node $v$ in the clique stays black since it has $k-1$ black neighbors and $k-1\ge \alpha d(v)$, which we prove below. Thus, it has a dynamo of size $k$.
\begin{align*}
& \alpha d(v)=\alpha (k-1+\frac{n}{k}-1)=\alpha(\sqrt{\frac{\alpha}{1-\alpha}n}+\sqrt{\frac{1-\alpha}{\alpha}n}-2)=\\
& \sqrt{\frac{\alpha}{1-\alpha}}(\alpha\sqrt{n}+(1-\alpha)\sqrt{n})-2\alpha\stackbin{\alpha\ge 1/2}{\le} \sqrt{\frac{\alpha}{1-\alpha}n}-1=k-1.
\end{align*}
For $\alpha> \frac{1}{2}$ and the cycle $C_n$, a dynamo must include all nodes. Let $\mathcal{C}$ be a configuration on $C_n$ with at least one white node, in two-way $\alpha$-BP for $\alpha>\frac{1}{2}$ in the next round both its neighbors will be white. Thus, a configuration with one or more white nodes never reaches the fully black configuration. This implies that our trivial upper bound of $n$ is tight.

Now, we provide the following observation for the cycle $C_n$, which implies that the lower bound in part (ii) is tight for $\alpha\le\frac{1}{2}$. However, for $\frac{1}{2}<\alpha\le \frac{3}{4}$ we do not know whether this bound is tight or not.

\begin{observation}
\label{obs1}
$\overleftarrow{MD}_{\alpha}(C_n)=1$ for $\alpha\le \frac{1}{2}$ and odd $n$. 
\end{observation}
\begin{proof}
Consider an odd cycle $v_1,v_2,\cdots,v_{2k+1},v_1$ for $n=2k+1$. Assume that in the initial configuration $\mathcal{C}_0$ there is at least one black node, say $v_1$. In configuration $\mathcal{C}_1$, nodes $v_{2}$ and $v_{2k+1}$ are both black because $\alpha\le \frac{1}{2}$. With a simple inductive argument, after $k$ rounds two adjacent nodes $v_{k+1}$ and $v_{k+2}$ will be black. In the next round, they both stay black and nodes $v_{k}$ and $v_{k+3}$ get black as well. Again with an inductive argument, after at most $k$ more rounds all nodes will be black. \qed 
\end{proof}

\begin{theorem} \cite{reichman2012new}
\label{thm3}
For a graph $G$, $r\le MD_r(G)\le (\frac{r}{1+\delta})n$. 
\end{theorem}
The upper bound is known by prior work~\cite{reichman2012new}. The lower bound is trivial since if a configuration includes less than $r$ black nodes in $r$-BP, no white node will turn black in the next round. Furthermore, $MD_r(K_n)=r$ (see Lemma~\ref{lemma}), which implies that the lower bound and upper bound are both tight (note for $K_n$, $(\frac{r}{1+\delta})n=r$). 

\begin{theorem}
\label{thm4}
In a graph $G$, 

(i) if $r\ge 2$, $r\le \overleftarrow{MD}_r(G)\le n$

(ii) if $r=1$, $\overleftarrow{MD}_r(G)=2$ if $G$ is bipartite and $\overleftarrow{MD}_r(G)=1$ otherwise.
\end{theorem}
\begin{proof}
The bounds in part (i) are trivial. 

For part (ii), let us first discuss that any two adjacent nodes in $G$ are a dynamo in two-way $1$-BP. Let $v$ and $u$ be two adjacent nodes and assume that the process triggers from a configuration in which $v$ and $u$ are black. A simple inductive argument implies that in the $t$-th round for $t\ge 0$ all nodes in distance at most $t$ from $v$ (similarly $u$) are black. Thus, after $t'$ rounds for some $t'$ smaller than the diameter of $G$ the graph is fully black. 


Now, we prove that $\overleftarrow{MD}_r(G)=2$ for $r=1$ if $G$ is bipartite. From above, we know that $\overleftarrow{MD}_r(G)\le 2$; thus, it remains to show that $\overleftarrow{MD}_r(G)\ge 2$. We argue that a configuration with only one black node cannot make $G$ fully black. Since $G$ is bipartite, we can partition the node set of $G$ into two non-empty independent sets $U$ and $W$. Without loss of generality, assume that we start from a configuration where a node in $U$ is black and all other nodes are white. Since all neighbors of the nodes in $U$ are in $W$ and vice versa, the color of each node in $U$ in round $t$ is only a function of the color of nodes in $W$ in the $(t-1)$-th round and the other way around. This means by starting from such a configuration, in the next round all nodes in $U$ will be white because all nodes in $W$ are white initially. In the round after that, all nodes in $W$ will be white with the same argument and so on. By an inductive argument, $\mathcal{C}_t|_U=w$ for odd $t$ and $\mathcal{C}_t|_W=w$ for even $t$. Thus, there is no dynamo of size $1$.

Finally, we prove that if $G$ is non-bipartite, then it has a dynamo of size $1$. Since $G$ is not bipartite, it has at least one odd cycle. Let $C_n:=v_1,v_2,\cdots,v_{2k+1},v_1$ be an arbitrary odd cycle in $G$ of size $2k+1$ for some integer $k\ge 1$. Now, suppose that the process starts from a configuration where node $v_1$ is black. In the next round nodes $v_2$ and $v_{2k+1}$ will be black. After one more round nodes $v_3$ and $v_{2k}$ will be black. By applying the same argument, after $k$ rounds nodes $v_{k+1}$ and $v_{k+2}$ will be black. As we discussed above, two adjacent black nodes make the graph fully black. Thus, in a non-bipartite graph, a node which is on an odd cycle is a dynamo of size one in two-way $1$-BP.  \qed
\end{proof}

\textbf{Tightness.} The bounds in part (i) are tight because $\overleftarrow{MD}_r(K_n)=r$ and $\overleftarrow{MD}_r(G)=n$ for any $r$-regular graph $G$ and $r\ge 2$ (see Lemma~\ref{lemma}).

For graphs $G=(V,E)$ and $G'=(V,E')$, if $E\subset E'$ then $MD_r(G')\le MD_r(G)$ and $\overleftarrow{MD}_{r}(G')\le \overleftarrow{MD}_{r}(G)$. This is true because by a simple inductive argument any dynamo in $G$ is also a dynamo in $G'$. Thus, if we keep adding edges to any graph, eventually it will have a dynamo of minimum possible size, namely $r$, in both $r$-BP and two-way $r$-BP. Thus, it would be interesting to ask for the degree-based density conditions that ensure that a graph has a dynamo of size $r$. Gunderson~\cite{gunderson2017minimum} proved that if $\delta \ge \frac{n}{2}+r$ for a graph $G$ ($r$ can be replaced by $r-3$ for $r\ge 4$), then $MD_r(G)=r$. We provide similar results for the two-way variant.
\begin{theorem}
\label{thm5}
If $\delta\ge \frac{n}{2}+r$ for a graph $G=(V,E)$, then it has $\Omega(n^r)$ dynamos of size $r$ in two-way $r$-BP. 
\end{theorem}
Note that this statement is stronger than Gunderson's result. Firstly, we prove that there is a dynamo of size $r$ in two-way $r$-BP, which immediately implies that there is a dynamo of such size in $r$-BP. In addition, we prove that actually there exist $\Omega(n^r)$ of such dynamos (this is asymptotically the best possible since there are ${n\choose r}=\mathcal{O}(n^r)$ sets of size $r$). It is worth to mention that our proof is substantially shorter.
\begin{proof}
For node subset $D$, define the sets $V_1^{D}:=\{v\in V:d_D(v)\ge r\}$ and $V_2^{D}:=V\setminus V_1^{D}$. Firstly, we have
\begin{equation}
\label{eq 4}
\sum_{v\in V}d_D(v)=\sum_{v\in D}d(v)\ge |D|\delta.
\end{equation}
Furthermore, since each node in $V_1^D$ has at most $|D|$ neighbors in $D$ and each node in $V_2^D$ has at most $(r-1)$ neighbors in $D$, we have
\begin{equation}
\label{eq 5}
\sum_{v\in V}d_D(v)\le |D|\cdot |V_1^D|+(r-1)|V_2^D|=(|D|-(r-1))|V_1^D|+(r-1)n
\end{equation}
where we apply $|V_1^D|+|V_2^D|=n$.
By combining Equations~\ref{eq 4} and~\ref{eq 5} and using $\delta\ge \frac{n}{2}+r$, we have $(|D|-(r-1))|V_1^D|+(r-1)n\ge \frac{n}{2}|D|+r|D|$. Dividing both sides by $|D|$ and rearranging the terms give us
\begin{equation}
\label{eq 3}
|V_1^D|\ge \frac{\frac{1}{2}-\frac{r-1}{|D|}}{1-\frac{r-1}{|D|}}n+\frac{r}{1-\frac{r-1}{|D|}}\xRightarrow {1-\frac{r-1}{|D|}\le 1} |V_1^D|\ge \frac{|D|-2r+2}{2|D|-2r+2}n+r.
\end{equation}
Building on Equation~\ref{eq 3}, we prove that any node set of size $2r-1$ has at least one subset of size $r$ which is a dynamo in the two-way $r$-BP. There are ${n\choose 2r-1}=\Omega(n^{2r-1})$ sets of size $2r-1$ and a set of size $r$ is shared by ${n-r\choose r-1}=\mathcal
{O}(n^{r-1})$ sets of size $2r-1$. Thus, there exist $\Omega(n^{2r-1})/\mathcal{O}(n^{r-1})=\Omega(n^r)$ distinct dynamos of size $r$.

Let $D_1\subset V$ be an arbitrary set of size $2r-1$. By setting $D=D_1$ in Equation~\ref{eq 3} and applying $|D_1|=2r-1$, we have $|V_1^{D_1}|\ge \frac{n}{2r}$. Recall that $V_1^{D_1}$ are the nodes which have at least $r$ neighbors in $D_1$. By the pigeonhole principle, there is a subset $D_2\subset D_1$ of size $r$ so that at least $(n/2r)/{2r-1\choose r}\ge n/(2r)^{2r}$ nodes each have at least $r$ neighbors in $D_2$. We want to prove that $D_2$ is a dynamo. If initially $D_2$ is fully black, in the next round at least $n/(2r)^{2r}$ nodes will be black. Now, we prove that if an arbitrary set $D_3$ of size at least $n/(2r)^{2r}$ is black, the whole graph gets black in at most three more rounds.

By setting $D=D_3$ in Equation~\ref{eq 3}, $|V_1^{D_3}|\ge \frac{|D_3|-2r+2}{2|D_3|}n=\frac{n}{2}-\frac{(r-1)(2r)^{2r}}{n}n\ge \frac{n}{2}-(2r)^{3r}$, which is the number of black nodes generated by $D_3$. Let $D_4$ be a set of size at least $\frac{n}{2}-(2r)^{3r}$. We show that $|V_1^{D_4}|\ge \frac{n}{2}$, which means if $\mathcal{C}_t|_{D_4}=b$ for some $t\ge 0$, there will be at least $\frac{n}{2}$ black nodes in $\mathcal{C}_{t+1}$. Again applying Equation~\ref{eq 3} gives us
\begin{align*}
|V_1^{D_4}|\ge \frac{|D_4|-2r+2}{2|D_4|-2r+2}n+r=\frac{\frac{n}{2}-(2r)^{3r}-2r+2}{n-2(2r)^{3r}-2r+2}n+r\ge (\frac{n}{2}-r)+r=\frac{n}{2}.
\end{align*}
The last inequality holds because we have
\begin{align*}
&\frac{\frac{n}{2}-(2r)^{3r}-2r+2}{n-2(2r)^{3r}-2r+2}n\ge \frac{n-2r}{2}\Leftrightarrow n^2-2(2r)^{3r}n-4rn+4n\ge\\
& n^2-2(2r)^{3r}n-2rn+2n-2rn+4r(2r)^{3r}+4r^2-4r\Leftrightarrow 2n\ge 4r(2r)^{3r}+4r^2-4r.
\end{align*}
(Notice that $2n\ge 4r(2r)^{3r}+4r^2-4r$ is true since $r$ is fixed while $n$ tends to infinity.)

Finally, we prove that if $n/2$ nodes are black in some configuration, the whole graph gets black in the next round. Consider an arbitrary node $v$. Since $d(v)\ge \frac{n}{2}+r$ and there are at least $n/2$ black nodes, $v$ has at least $r$ black neighbors and will be black in the next round. \qed 
\end{proof}
\subsection{Monotone Dynamos}
\label{monotone}
Let us first define a monotone dynamo formally. For a graph $G=(V,E)$, we say a node set $D$ is a \emph{monotone dynamo} whenever the following holds: If $\mathcal{C}_0|_D=b$ and $\mathcal{C}_0|_{V\setminus D}=w$, then for some $t\ge 1$ we have $\mathcal{C}_t|_V=b$ and $\mathcal{C}_{t'-1}\le\mathcal{C}_{t'} $ for any $t'\le t$, which means any black node in $\mathcal{C}_{t'-1}$ is also black in $\mathcal{C}_{t'}$. Now, we provide bounds on the minimum size of a monotone dynamo. Since a dynamo in $r$-BP and $\alpha$-BP is monotone by definition, our bounds from Section~\ref{dynamo} apply.

For a graph $G=(V,E)$, the minimum size of a monotone dynamo in two-way $r$-BP is lower-bounded by $r+1$. Assume that there is a monotone dynamo $D$ of size $r$ or smaller. If $\mathcal{C}_0|_D=b$ and $\mathcal{C}_0|_{V\setminus D}=w$, then $\mathcal{C}_1|_D=w$; this is in contradiction with the monotonicity of $D$. This lower bound is tight because in $K_n$, a set of size $r+1$ is a monotone dynamo. Furthermore, the trivial upper bound of $n$ is tight for $r$-regular graphs with $r\ge 2$ (see Lemma~\ref{lemma}). For $r=1$, any two adjacent nodes are a monotone dynamo; thus, the minimum size of a monotone dynamo is two.

In two-way $\alpha$-BP, for $\alpha\le\frac{1}{2}$, on the cycle $C_n$ any two adjacent nodes are a monotone dynamo, which provides the tight lower bound of $2$. However, we are not aware of any non-trivial upper bound. For $\alpha>\frac{1}{2}$, the trivial upper bound of $n$ is tight for $C_n$. We provide the lower bound of $\sqrt{\frac{\alpha}{1-\alpha}n}-1$ in Theorem~\ref{thm6}, which is tight since the construction given for the tightness of Theorem~\ref{thm2} provides a monotone dynamo whose size matches our lower bound, up to an additive constant. Furthermore, we show that if we restrict the underlying graph to be a tree, we get the stronger bound of $\frac{\alpha}{2-\alpha}n$.   
\begin{theorem}
\label{thm6}
For a graph $G=(V,E)$ and two-way $\alpha$-BP with $\alpha>\frac{1}{2}$, the minimum size of a monotone dynamo is at least $\sqrt{\frac{\alpha}{1-\alpha}n}-1$, and at least $\frac{\alpha}{2-\alpha}n$ if $G$ is a tree.
\end{theorem}
\begin{proof}
Let set $D\subseteq V$ be a monotone dynamo in $G$. Suppose the process starts from the configuration where only all nodes in $D$ are black. Let $D_t$ denote the set of black nodes in round $t$. Then, $D_0=D$ and $D_t\subseteq D_{t+1}$ by the monotonicity of $D$. Furthermore, define the potential function $\Phi_t:=\partial(D_t)$. We claim that $\Phi_{t+1}\le \Phi_{t}-|D_{t}\setminus D_{t-1}|$ because for any newly added black node, i.e. any node in $D_{t}\setminus D_{t-1}$, the number of neighbors in $D_t$ is strictly larger than $V\setminus D_{t}$ (note that $\alpha>\frac{1}{2}$). In addition, since $D$ is a dynamo, $\mathcal{C}_T|_V=b$ for some $T\ge 0$, which implies $\Phi_T=0$. Thus, 
\begin{equation}
\label{eq 1}
\Phi_T=0\le \Phi_0-(n-|D|)\Rightarrow n\le \Phi_0+|D|.
\end{equation}
For $v\in D$, $d_{V\setminus D}(v)\le \frac{1-\alpha}{\alpha}d_{D}(v)$ because $D$ is a monotone dynamo and at least $\alpha$ fraction of $v$'s neighbors must be in $D$. Furthermore, $d_{D}(v)\le |D|-1$, which implies that $d_{V\setminus D}(v)\le \frac{1-\alpha}{\alpha}(|D|-1)$. Now, we have
\begin{equation}
\label{eq 2}
\Phi_0=\partial(D)=\sum_{v\in D}d_{V\setminus D}(v)\le \frac{1-\alpha}{\alpha}\sum_{v\in D} (|D|-1)=\frac{1-\alpha}{\alpha}|D|^2-\frac{1-\alpha}{\alpha}|D|.
\end{equation}
Putting Equations~\ref{eq 1} and ~\ref{eq 2} in parallel, plus some small calculations, imply that
\[
n\le\frac{1-\alpha}{\alpha}|D|^2+(1-\frac{1-\alpha}{\alpha})|D| \Rightarrow \sqrt{\frac{\alpha}{1-\alpha}n}-1\le |D|.
\]
When $G$ is a tree, we have
\[
\Phi_0=\partial(D)=\sum_{v\in D}d_{V\setminus D}(v)\le \frac{1-\alpha}{\alpha}\sum_{v\in D} d_{D}(v)\le \frac{1-\alpha}{\alpha}\ 2(|D|-1)
\]
because the induced subgraph by $D$ is a forest with $|D|$ nodes, which thus has at most $|D|-1$ edges. Combining this inequality and Equation~\ref{eq 1} yields
\[
n\le \frac{2(1-\alpha)}{\alpha}|D|-2+|D|\Rightarrow n\le \frac{2-\alpha}{\alpha}|D|\Rightarrow \frac{\alpha}{2-\alpha}n \le |D|.  
\] \qed
\end{proof}

\subsection{Stable and Immortal Sets}
\label{stable}
In this section, we provide tight bounds on the minimum size of a stable/immortal set. In $\alpha$-BP and $r$-BP a black node stays unchanged, which simply implies that $MS_{\alpha}(G)=MI_{\alpha}(G)=MS_r(G)=MI_r(G)=1$ for a graph $G$. Thus, we focus on the two-way variants in the rest of the section. The bounds are given in Table~\ref{Table 2}.
\begin{table}
\centering
\caption{The minimum size of a stable/immortal set. All our bounds are tight up to an additive constant. $x=0$ and $x=1$ respectively for odd and even $n$.}
\label{Table 2}
\begin{tabular}{|c|c|c|c|c|}
 \hline
      \multirow{2}{*}{Model} &
      \multicolumn{2}{c|}{Stable Set} &
      \multicolumn{2}{c|}{Immortal Set} \\
      \cline{2-5}
& Lower Bound & Upper Bound& Lower Bound & Upper Bound\\
 \hline
  Two-way $\alpha$-BP $\ \alpha\le\frac{1}{2}$&   $\lceil \frac{1}{1-\alpha}\rceil$  & $\alpha n$ & $1$ & $\alpha n$\\
  \hline
  Two-way $\alpha$-BP $\ \alpha>\frac{1}{2}$&   $\lceil \frac{1}{1-\alpha}\rceil$  & $n$ & $1$ & $n$\\
  \hline
 Two-way $r$-BP$\ r=1$&   $2$  &$2$ & $1$ & $1$\\
 \hline
  Two-way $r$-BP$\ r=2$&   $r+1$  &$n$ & $2$ & $\frac{n}{1+x}$\\
  \hline
 Two-way $r$-BP $\ r\ge 3$ & $r+1$ & $n$ & $r$ & $n$\\
 \hline
\end{tabular}
\end{table}

\textbf{Stable sets in two-way $\alpha$-BP.} We present tight bounds on $\overleftarrow{MS}_{\alpha}(G)$ in Theorem~\ref{thm7}.
\begin{theorem}
\label{thm7}
For a graph $G=(V,E)$,

(i) $\lceil\frac{1}{1-\alpha}\rceil\le\overleftarrow{MS}_{\alpha}(G)\le n$ for $\alpha>\frac{1}{2}$

(ii) $2=\lceil\frac{1}{1-\alpha}\rceil\le\overleftarrow{MS}_{\alpha}(G)\le \alpha n+\mathcal{O}(1)$ for $\alpha\le\frac{1}{2}$.
\end{theorem}

\begin{proof}
To prove the lower bound of $\lceil\frac{1}{1-\alpha}\rceil$, let the node set $S\subseteq V$ be a stable set in two-way $\alpha$-BP. Since $G$ is connected, there is a node $v\in S$ that shares at least one edge with $V\setminus S$, which implies that $d(v)\ge d_S(v)+1$. Furthermore, $d_S(v)\ge \alpha d(v)$ because $S$ is stable. Hence, $d_S(v)\ge \alpha (d_S(v)+1)$ which yields $d_S(v)\ge \frac{\alpha}{1-\alpha}$. Moreover, $|S|-1\ge d_S(v)$, which implies that $|S|\ge \frac{\alpha}{1-\alpha}+1=\frac{1}{1-\alpha}$. As $|S|$ is a positive integer, $|S|\ge \lceil\frac{1}{1-\alpha}\rceil$.

Let us prove the upper bound of $\alpha n+\mathcal{O}(1)$. For simplicity assume that $\alpha=1/c$ for some positive integer $c$ (the general case can be proven analogously). Partition the node set $V$ into subsets $V_1,\cdots,V_c$ such that each $V_i$ for $1\le i\le c$ includes at least $\lfloor n/c\rfloor-1$ nodes and the number of edges between $V_i$s is minimized. Since each subset contains at least $\lfloor n/c\rfloor-1$ nodes, the largest subset, which we denote by $V_{\text{max}}$, has at most $n/c+2c=\alpha n+\mathcal{O}(1)$ nodes. We claim that for each node $v\in V_{\text{max}}$, $d_{V_{\text{max}}}(v)\ge \frac{d(v)}{c}=\alpha d(v)$, which implies that $V_{\text{max}}$ is a stable set. Assume that there is a node $u\in V_{\text{max}}$ which violates this property. Thus, $d_{V\setminus V_{\text{max}}}(u)>(1-\frac{1}{c})d(u)=\frac{c-1}{c}d(u)$. By  the pigeonhole principle, there must exist a subset $V'$ among the $c-1$ remaining subsets such that $d_{V'}(u)>\frac{1}{c}d(u)$. This is a contradiction because by removing $u$ from $V_{\text{max}}$ and adding it into $V'$, the number of edges between the subsets decreases at least by one. Notice that the new subsets stratify our desired size constraints because $V_{\text{max}}$ is originally of size at least $\lfloor n/c\rfloor$ and by removing a node from it, its size will not be less than $\lfloor n/c\rfloor-1$. \qed
\end{proof}

\textbf{Tightness.} The bounds in Theorem~\ref{thm7} are tight, up to an additive constant. For th lower bound, consider graph $G=(V_1\cup V_2,E)$, where $|V_1|=n-\lceil \frac{1}{1-\alpha}\rceil$ and $|V_2|=\lceil \frac{1}{1-\alpha}\rceil$. Assume that the nodes in $V_2$ create a clique, the nodes in $V_1$ create an arbitrary connected graph, and finally there is an edge between node $v\in V_2$ and node $v'\in V_1$. $G$ is an $n$-node connected graph. $V_2$ is a stable set for $G$ in two-way $\alpha$-BP because firstly for each node $u\in V_2\setminus\{v\}$, $d_{V_2}(u)=d(u)\ge \alpha d(u)$. In addition for node $v$,
\[
(1-\alpha)\lceil \frac{1}{1-\alpha}\rceil\ge 1\Rightarrow \lceil \frac{1}{1-\alpha}\rceil-1\ge \alpha \lceil \frac{1}{1-\alpha}\rceil\Rightarrow d_{V_2}(v)\ge \alpha d(v)
\] 
where we used $d_{V_2}(v)=\lceil \frac{1}{1-\alpha}\rceil-1$ and $d(v)=\lceil \frac{1}{1-\alpha}\rceil$. 

Our upper bounds are also tight because $\overleftarrow{MS}_{\alpha}(C_n)=n$ for $\alpha>\frac{1}{2}$ and $\overleftarrow{MS}_{\alpha}(K_n)\ge\alpha (n-1)$.

\textbf{Stable sets in two-way $r$-BP.} For a graph $G$, $\overleftarrow{MS}_r(G)=2$ for $r=1$ because two adjacent nodes create a stable set. For $r\ge 2$, we have tight bounds of $r+1\le\overleftarrow{MS}_r(G)\le n$. Notice if in a configuration less than $r+1$ nodes are black, in the next round all black nodes turn white. Furthermore, the lower bound of $r+1$ is tight for $K_n$ and the upper bound is tight for $r$-regular graphs.

\textbf{Immortal sets in two-way $\alpha$-BP.} For a graph $G$, $1\le\overleftarrow{MI}_{\alpha}(G)\le n$ for $\alpha>\frac{1}{2}$ and $1\le\overleftarrow{MI}_{\alpha}(G)\le \alpha n+\mathcal{O}(1)$ for $\alpha\le\frac{1}{2}$. All bounds are trivial except $\alpha n+\mathcal{O}(1)$, which is a corollary of Theorem~\ref{thm7} (note that stability implies immortality). The lower bound of 1 is tight since the internal node of the star graph $S_n$ is an immortal set of size 1. Regarding the tightness of the upper bounds, we have $\overleftarrow{MI}_{\alpha}(K_n)\ge \alpha(n-1)$ for $\alpha\le 1/2$ and $\overleftarrow{MI}_{\alpha}(C_n)=n$ for $\alpha>1/2$ and odd $n$ (this basically follows from the proof of Observation~\ref{obs1} by replacing black with white and $\alpha\le\frac{1}{2}$ with $\alpha>\frac{1}{2}$). 

\textbf{Immortal sets in two-way $r$-BP.} We provide tight bounds on $\overleftarrow{MI}_r(G)$ in Theorem~\ref{thm8}. Interestingly, the parity of $n$ plays a key role concerning the minimum size of an immortal set for $r=2$.
\begin{theorem}
\label{thm8}
For a graph $G=(V,E)$, 

(i) if $r=1$, $\overleftarrow{MI}_r(G)=1$

(ii) if $r=2$, $2\le \overleftarrow{MI}_r(G)\le \frac{n}{1+x}$, where $x=0$ for odd $n$ and $x=1$ for even $n$.

(iii) if $r\ge 3$, $r\le \overleftarrow{MI}_r(G)\le n$.
\end{theorem}
\begin{proof}
For $r=1$, any node set of size 1 is immortal because if a node $v\in V$ is black in some configuration in two-way $1$-BP, in the next round all nodes in $\Gamma(v)$ will be black.

The lower bound of $r$ is trivial because a configuration with less than $r$ black nodes in two-way $r$-BP gets fully white in the next round. 

All the upper bounds are also trivial, except the bound $n/2$. Assume $n$ is even and $r=2$; we prove that $G$ has an immortal set of size at most $n/2$. Let
$C_k:=u_1,u_2,\cdots, u_i,u_{i+1},\cdots, u_k,u_1$ 
be a cycle of length $k$ in $G$, then the node set $U:=\{u_1,\cdots,u_k\}$ is a stable set, and consequently an immortal set. This is true because each node has $r=2$ neighbors in the set; that is if all nodes in $U$ are initially black, they stay black forever. If $k$ is even, then actually the smaller set $U_{\text{e}}:=\{u_i\in U:i\ \text{is even} \}$ is an immortal set. This is correct because in two-way $2$-BP if for some configuration all nodes in $U_{\text{e}}$ are black, in the next round all nodes in $U\setminus U_{\text{e}}$ will be black, irrespective of the color of the other nodes. One round later, all nodes in $U_{\text{e}}$ will be black again and so on. 

Now, let $C$ of length $k$ be a longest cycle in $G$. If $k$ is even then, as discussed above, half of the nodes in $C$ suffice to create an immortal set. If $k\le n/2$, then the nodes on $C$ create an immortal set. In both cases there is an immortal set of size at most $n/2$. Thus, assume otherwise, i.e., $k$ is larger than $n/2$ and odd. Let the node set $V_1$ include all nodes in $C$ and $V_2:=V\setminus V_2$. Since the graph is connected and $V_2\ne \emptyset$ (we already excluded the case of $k= n$), there is a node $w_1\in V_2$  and a node $u\in V_1$ such that $\{w_1,u\}\in E$. We start traversing from $w_1$. Since $r=2$, each node, including $w_1$, is of degree at least $2$ (recall that we always assume that $\delta(G)\ge r$, see Section~\ref{pre}). Thus, $w_1$ in addition to $u$ has at least another neighbor, say $w_2$. Since $w_2$ is of degree at least $2$ as well, it must be adjacent to another node, say $w_3$, and so on. Assume that $w_{k'}$ for some $k'\ge 2$ is the first node from $V_1$, which we visit during the traversing. The union of path $u, w_1, w_2,\cdots, w_{k'}$ and cycle $C$ gives us two new cycles if $u\ne w_{k'}$. Since the length of cycle $C$ is odd, one of these two new cycles must be of even length (see Fig.~\ref{fig1}). This even cycle provides us with an immortal set of size $n/2$. If $u=w_{k'}$ then we have a cycle which includes only nodes from $V_2$ plus $u$. This cycle is of size at most $n/2$ because $|V_2|=|V|-|V_1|=n-k< n/2$ by applying $k>n/2$. This gives an immortal set of size at most $n/2$. If we never visit a node from $V_1$, then we eventually revisit a node $w_j\in V_2$ which gives us a cycle on some nodes in $V_2$. Notice this cycle is of size at most $n/2$ since $|V_2|<n/2$.

\begin{figure}[h]
\centering
\includegraphics{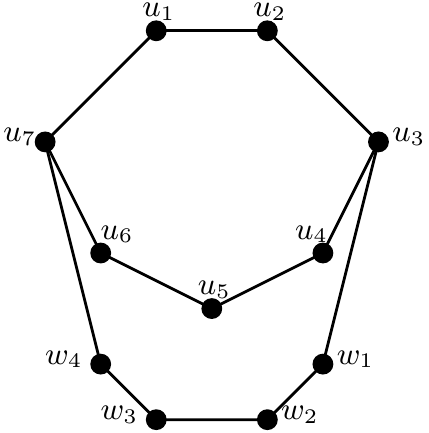}
\caption{By connecting each of the endpoints of a path into a different node on a cycle, two new cycles are generated. If the original cycle is of odd size, here $7$, then one of the two new cycles is of even size, here the cycle $u_1,u_2,u_3,w_1,\cdots,u_7,u_1$}
\label{fig1}
\end{figure}
 \qed
\end{proof}

\textbf{Tightness.} The lower bounds are tight because $\overleftarrow{MI}_r(K_n)=r$ for $r\ge 1$. The upper bounds for $r=2$ are tight because $\overleftarrow{MI}_2(C_n)$ is equal to $n$ for odd $n$ and it is equal to $n/2$ for even $n$. 

Regarding the tightness of the upper bound of $n$ for $r\ge 3$, we prove that for any sufficiently large $n$, there is an $n$-node graph $G=(V,E)$ which has no immortal set of size smaller than $n-2r-1$. Let $n$ be even; we will discuss how our argument applies to the odd case. We first present the construction of graph $G$ step by step. For $1\le i\le K:=\lfloor n/(r+1)\rfloor-1$, let $G_i$ be the clique on the node set $V_i:=\{v_i^{(j)}:1\leq j\leq r+1\}$ minus the edge $\{v_i^{(1)},v_i^{(2)}\}$. To create the first part of graph $G$, we connect $G_i$s with a path. More precisely, we add the edge set $\{\{v_i^{(2)},v_{i+1}^{(1)}\}: 1\le i\le K-1\}$. So far the generated graph is $r$-regular except the nodes $v_1^{(1)}$ and $v_K^{(2)}$ which are of degree $r-1$ and we have $\ell:=n-K(r+1)$ nodes left. Let $G'$ be an arbitrary $r$-regular graph on $\ell$ nodes. Now, remove an edge $\{v',u'\}$ from $G'$ and connect $v'$ to $v_1^{(1)}$ and $u'$ to $v_{K}^{(2)}$. Clearly the resulting graph $G$ is $r$-regular with $n$ nodes (see Fig.~\ref{fig2} for an example). However, we should discuss that such a graph $G'$ exists. The necessary and sufficient condition for the existence of an $r$-regular graph on $\ell$ nodes is that $\ell\ge r+1$ and $r\ell$ is even. Firstly, $\ell=n-(\lfloor n/(r+1)\rfloor-1)(r+1)\ge r+1$. Furthermore, if $r$ is even, then $r\ell$ is even; thus, assume otherwise. Since $r$ is odd and $n$ is even, $rn$ is even. In addition, $r(r+1)$ is even, which implies that $r(r+1)(\lfloor n/(r+1)\rfloor-1)$ is even. Overall, $r\ell=rn-r(r+1)(\lfloor n/(r+1)\rfloor-1)$ is even. 

\begin{figure}[h]
\centering
\includegraphics{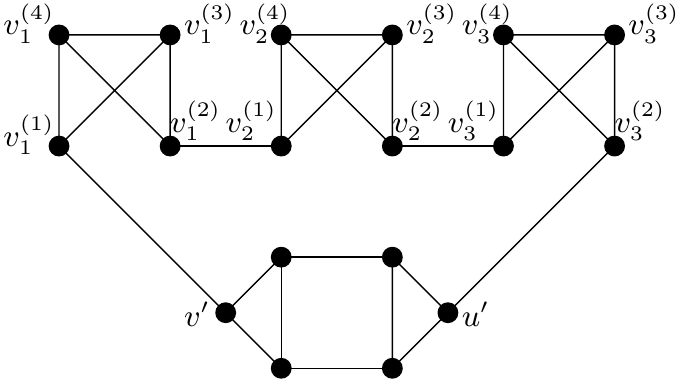}
\caption{The construction of graph $G$ for $r=3$ and $n=18$.}
\label{fig2}
\end{figure}
Now, we claim that all nodes in $V_i$ for $1\le i\le K$ must be in any immortal set. This implies that the minimum size of an immortal set is at least $K(r+1)=(\lfloor n/(r+1)\rfloor-1)(r+1)\ge n-2r-1$. Assume that we start from a configuration in which all nodes are black except a node $v$ in one of $V_i$s. If $v$ is $v_i^{(1)}$ or $v_i^{(2)}$, then in the next round $v_i^{(3)}$ and $v_i^{(4)}$ (which must exist because $r\ge 3$) are both white since each of them has at most $r-1$ black neighbors. By construction, there is an edge between $v_i^{(3)}$ and $v_i^{(4)}$. Since $G$ is $r$-regular, in the next round they both stay white and all their neighbors get white as well. By an inductive argument in the $t$-th round for $t\ge 1$, all nodes whose distance is at most $t$ from $v_i^{(3)}$ (or similarly $v_i^{(4)}$) will be white. Thus, eventually the graph is going to be fully white. Now, assume that $v$ is a node in $V_i\setminus\{v_i^{(1)},v_i^{(2)}\}$. Again, in the next round $v_i^{(1)}$ and one round after that $v_i^{(3)}$ and $v_i^{(4)}$ get white and the same argument follows. If $n$ is odd, we do the same construction for $n-1$ instead of $n$ and at the end, add a node $w$ and connect it to $v_{i}^{(2)}$ for $1\le i\le r$. This graph is not $r$-regular since there are $r$ nodes of degree $r+1$. However, a similar argument applies since from any configuration with two adjacent white nodes in one of $V_i$s, the whole graph gets fully white eventually. 

\bibliographystyle{splncs03}
\bibliography{refer}
\end{document}